\documentclass[11pt,journal,draftclsnofoot,onecolumn]{IEEEtran}

\usepackage{cite}
\usepackage{amsmath,amssymb,amsfonts,amsthm}
\usepackage{algorithmic}
\usepackage{graphicx}
\usepackage{textcomp}
\usepackage{xcolor}
\usepackage{url}
\usepackage{xspace}
\usepackage{booktabs}
\usepackage{mathrsfs}

\def\BibTeX{{\rm B\kern-.05em{\sc i\kern-.025em b}\kern-.08em
    T\kern-.1667em\lower.7ex\hbox{E}\kern-.125emX}}

\newcommand{\comment}[1]{}
\newcommand{\name}{$\mathcal{R} \mathcal{S} \mathcal{V}{\mathcal{P}}$\xspace}

\newtheorem{theorem}{Theorem}
\newtheorem{corollary}{Corollary}

\newcommand\circledcheck[1][green]{%
{\scalebox{1.5}{\textcolor{#1}{\textcircled{\normalcolor$\checkmark$}}}}%
}

\newcommand\circledcross[1][red]{%
{\scalebox{1.5}{\textcolor{#1}{\textcircled{\normalcolor$\times$}}}}%
}

\begin{document}

\title{\name: Beyond Weisfeiler Lehman Graph Isomorphism Test}
\author{}

\author{\IEEEauthorblockN{Sourav Dutta} \\
\IEEEauthorblockA{\textit{Huawei Ireland Research} \\
Dublin, Ireland \\
Email: sourav.dutta2@huawei.com} \\
\and
\IEEEauthorblockN{Arnab Bhattacharya} \\
\IEEEauthorblockA{\textit{Indian Institute of Technology} \\
Kanpur, India \\
Email: arnabb@cse.iitk.ac.in}
}

\maketitle

\begin{abstract}
	\emph{Graph isomorphism}, a classical algorithmic problem, determines whether two input graphs are structurally identical or not. Interestingly, it is one of the few problems that is not yet known to 
	belong to either the {\em P} or {\em NP-complete} complexity classes. As such, intelligent search-space pruning based strategies were proposed for developing isomorphism testing solvers like 
	{\tt nauty} and {\tt bliss}, which are still, unfortunately, exponential in the worst-case scenario. Thus, the polynomial-time {\em Weisfeiler-Lehman} (WL) isomorphism testing heuristic, based on {\em colour 
	refinement}, has been widely adopted in the literature. More recently graph isomorphism has been used for graph neural network applications as well. However, WL fails for multiple classes of non-isomorphic graph instances such as 
	strongly regular graphs, block structures, and switched edges, among others.

	In this paper, we propose a novel \emph{polynomial-time} graph isomorphism testing heuristic, \name, and depict its enhanced discriminative power compared to the Weisfeiler-Lehman approach for several 
	challenging classes of graphs. Bounded by a run-time complexity of $\mathcal{O}(m^2 + mn^2 + n^3)$ (where $n$ and $m$ are the number of vertices and edges respectively), we show that \name can identify 
	non-isomorphism in several ``hard'' graph instance classes including {\em Miyazaki, Paulus, cubic 
	hypohamiltonian, strongly regular, Latin series} and {\em Steiner triple system graphs}, where the $3$-WL test fails. Algorithmically, for every vertex of an input graph, \name computes the {\em vertex signature} 
	based on {\em prime encoding} of ``reachability distance'' to other vertices in the graph (using breadth-first search). These vertex signatures coupled with the average pair-wise 
	distance among the parents of a vertex is then used to determine the isomorphism between the graphs. Similar to the WL test, our proposed algorithm is prone to only one-sided errors, where isomorphic graphs will never be determined to be non-isomorphic, although the reverse can happen. 
	Empirically, however, we demonstrate lower false positive rates (compared to WL) on multiple challenging graph instance classes, thereby improving the scope of non-isomorphism 
	detection to a broader class of graphs.
\end{abstract}


\newpage

\section{Introduction and Background}
\label{sec:intro}

{\bf Definition.} The {\em graph isomorphism} (GI) problem involves determining whether two graphs are structurally equivalent or identical~\cite{karp}. Mathematically, an isomorphism between two graphs is a bijection between the 
vertex sets of the graphs such that the \emph{adjacency} property is preserved. Isomorphism detection has several applications related to molecular graph matching in chemoinformatics~\cite{iso}, combinatorial optimizations in 
microprocessor circuits~\cite{micro}, computer vision~\cite{cv}, and more recently in graph neural networks~\cite{GNN_iso}. As such, this classical problem has been studied extensively since the early days of computing, both in terms of theoretical 
as well as practical point of view. GI was introduced by Karp in the seminal paper on the NP-completeness of combinatorial problems~\cite{karp}, and is still regarded as one of the major open problems in computer science. 
Intriguingly, although GI is {\em polynomial-time verifiable}, it is yet unknown whether it is solvable in polynomial time (i.e., complexity class $P$) or is $NP$-complete, or potentially in {\em intermediate 
complexity}~\cite{iso2,art}. However, NP-completeness is considered unlikely in the community, since it would imply the collapse of the polynomial-time hierarchy complexity classes above $NP$~\cite{np}.
Interestingly, the related problems of subgraph isomorphism has been shown to be $NP$-complete~\cite{cook} and approximate graph isomorphism to minimize the number of bijection mismatches is known to be $NP$-hard~\cite{Apprx_iso}.
Generally, in the above definition of isomorphism, the graphs are considered to be undirected non-labeled and non-weighted; however, the notion can be extended to the other variants as well. When the isomorphism of a 
graph is onto itself, it is referred to as {\em automorphism}.

{\bf Related Study.} The first important steps in terms of theoretically solving GI can be attributed to the {\em Whitney graph isomorphism theorem} stating that two connected graphs are isomorphic if and only if their 
line graphs are identical, with the sole exception being the complete graph on three vertices ($K_3$) and the complete bipartite graph $K_{1,3}$~\cite{whit}. Subsequently, an $\mathcal{O}(n \log n)$ algorithm was proposed for 
isomorphism testing for planar graphs~\cite{plane}. Polynomial time algorithms are known for certain other special classes of graphs as well, such as those with forbidden minors~\cite{pono, gro} and forbidden topological 
minors~\cite{gro2}, including graphs with bounded degrees, bounded genus and bounded tree-widths to name a few~\cite{master}. The ``na\"ive refinement'' procedure was shown to solve GI for a number of graph classes in linear 
time~\cite{gr1, gr2}. Group-theoretic methods along with the divide-and-conquer framework were proposed, leading to GI being solvable in $e^{\mathcal{O}(\sqrt{n \log n})}$ time~\cite{bound1a, bound1b} -- which remained the known 
bound for decades. 
Recently, Babai proposed a {\em quasi-polynomial time} GI algorithm~\cite{bound2a, bound2b} with a running time of $e^{(\log n)^{\mathcal{O}(1)}}$ (the best known bound so far), which was followed by other quasi-polynomial parameterized 
algorithms~\cite{survey_iso}. However, such algorithms are not practical for large graphs. A detailed overview of classical and recent techniques to tackle the GI problem can be found in~\cite{survey_iso}.

{\bf Intelligent Solvers.} To tackle GI in practice, several isomorphism testing softwares were developed including {\tt nauty}~\cite{nauty1, nauty2}, {\tt saucy}~\cite{saucy1, saucy2}, {\tt Traces}~\cite{traces1, traces2} 
and {\tt Bliss}~\cite{bliss1, bliss2}.
Such approaches mainly involve a combination of combinatorial strategies like {\em colour refinement} and {\em canonical labeling} along with intelligent use of data structures for efficiently pruning the search space. 
In fact, these approaches involve the innovative use of automorphism for pruning, and was shown to work quite well for large graphs as well as for a wide variety of graph classes~\cite{pract_iso, bench_iso}. Other efforts 
have studied probabilistic linear programming~\cite{LP_Iso} and eigenvalue splitting~\cite{eigen}. Experiments on benchmark graph datasets have shown that on inputs without any particular combinatorial structure, these solvers scale 
almost linearly~\cite{bench_iso}. Although it is non-trivial to create challenging instances for these solvers, these toolboxes work on the {\em individualized-refinement} paradigm~\cite{ind} and have a worst-case complexity bounded by 
an exhaustive search which is exponential~\cite{exp_nauty, exp_solver}. For example, earlier versions of {\em nauty} was shown to run in exponential time for the Miyazaki family of graphs (see Fig.~\ref{fig:examples}(c))~\cite{nauty_miya}.

\begin{figure}[t]
\centering
\resizebox{\columnwidth}{!}{
	\begin{tabular}{ccc}
		\includegraphics[width=0.3\columnwidth]{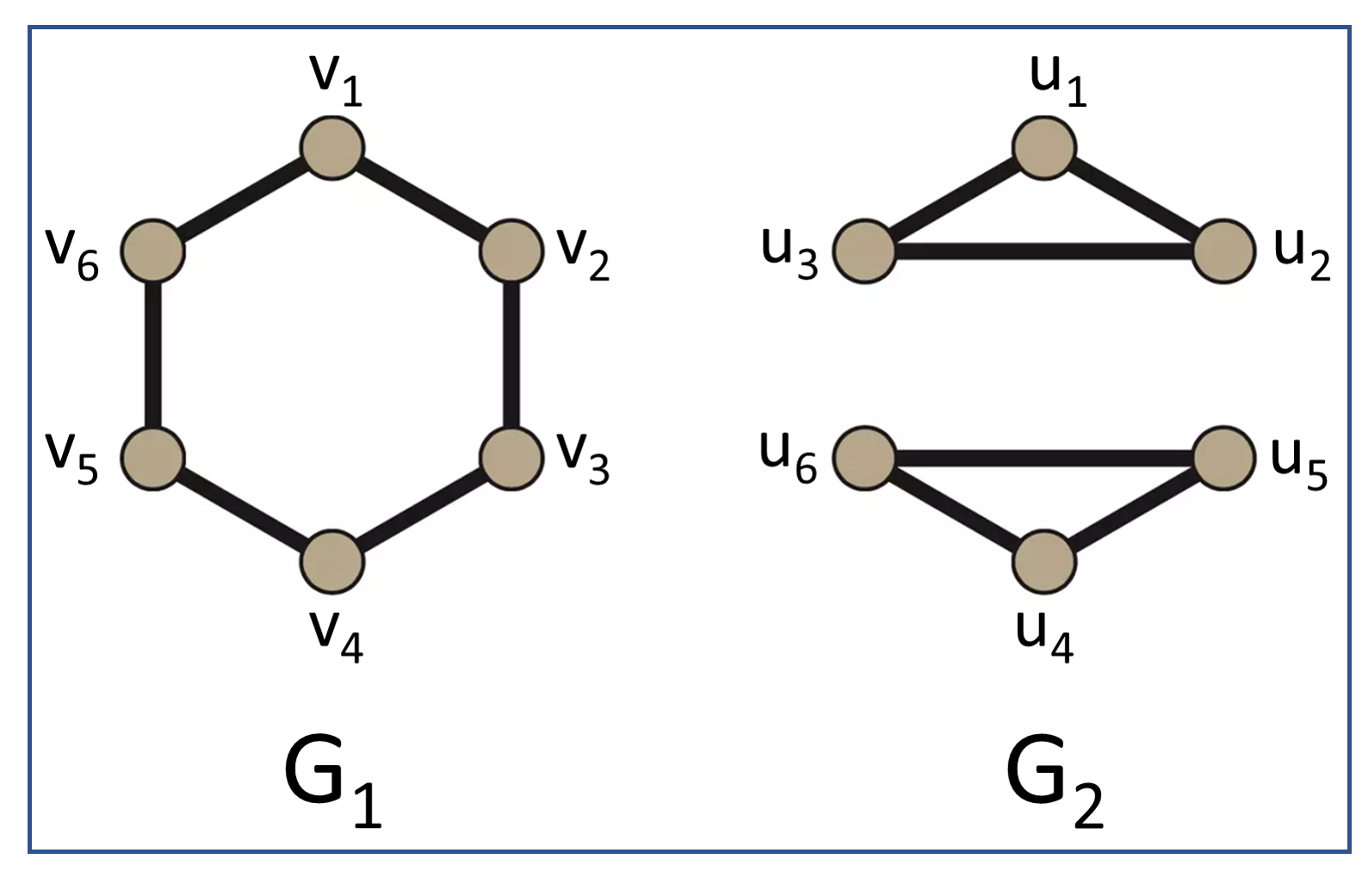} &
		\includegraphics[width=0.33\columnwidth]{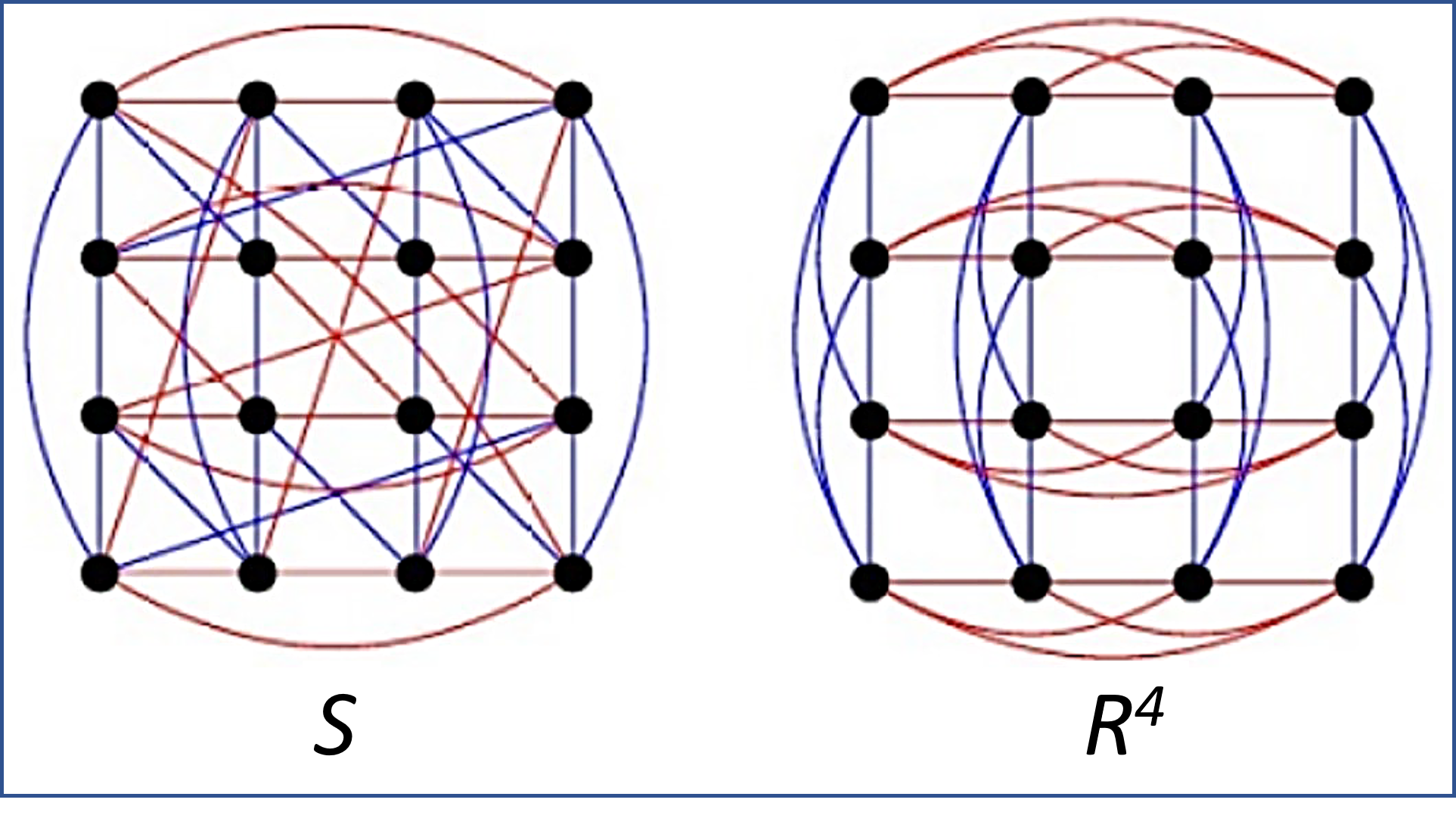} & 
		\includegraphics[width=0.58\columnwidth]{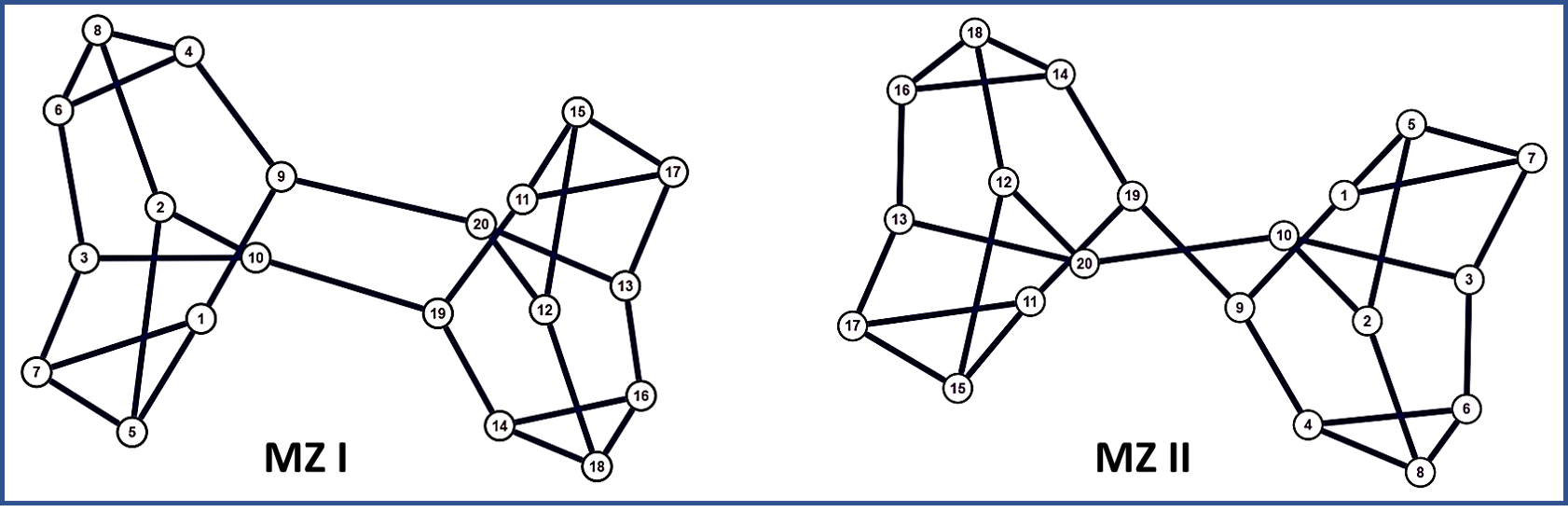} \\
		{\bf (a)} & {\bf (b)} & {\bf (c)}
	\end{tabular}}
	\caption{Non-isomorphic graph examples for which Weisfeiler-Lehman (WL) test fails: {\bf (a)} {\em Disconnected property} between 6-vertex graphs $G_1$ and $G_2$, {\bf (b)} {\em Strongly regular property} between Shrikhande 
	graph ($S$) and 4$\times$4 Rook's graph ($R^4$), and {\bf (c)} {\em Switched edge property} of Miyazaki Graphs (MZ I and MZ II). WL incorrectly reports the graphs pairs $G_1$:$G_2$, $S$:$R^4$, and MZ I:MZ II as isomorphic 
	based on colour refinement strategy (while the proposed \name heuristic correctly distinguishes them). Additionally, $3$-WL, $2$-FWL, and GD-WL fails to distinguish between $S$:$R^4$, while \name is successful.}
	\label{fig:examples}
\end{figure}

{\bf Polynomial Weisfeiler-Lehman Heuristic.} To this end, the well-known polynomial time {\em Weisfeiler-Lehman} (WL) algorithm is a classical isomorphism test based on color refinement~\cite{wl1, wl2}. 
The algorithm converges to provide 
``states'' or colours to the vertices based on their iterated degree sequence, thereby creating a representation or ``certificate'' of the graphs which are then compared for isomorphism. However, the WL test is ``incomplete'' 
and suffers 
from \emph{one-sided false positive} errors. In other words, if two graphs have different representations they are definitely non-isomorphic, but graphs with identical representations may not be isomorphic. In fact, the WL 
heuristic fails for certain seemingly easy disconnected graph instances as well as for other challenging graph families like strongly regular graphs, as shown in Fig.~\ref{fig:examples}(a) and Fig.~\ref{fig:examples}(b) respectively. To 
combat the above problem 
to a certain extent, the $k$-dimensional WL ($k$-WL)~\cite{kwl2, kwl3, kwl4} and $k$-folklore-WL ($k$-FWL)~\cite{kfwl} approaches have been proposed as natural generalizations to the original $1$-dimensional WL heuristic. The running time 
complexity of the $k$-dimensional WL test is polynomially bounded by $\mathcal{O}(n^{k+1} \log n)$~\cite{kwl1} for a constant $k$. Thus, WL has been widely adopted not only for fast isomorphism testing but also for modern applications such as 
graph kernels~\cite{ker}, graph neural networks (GNN)~\cite{GNN_iso}, graph property recognition~\cite{gp}, and link prediction~\cite{link}. In fact, the quasi-polynomial algorithm of Babai employs the $k$-WL heuristic for 
$k=\mathcal{\mathcal{O}}(\log n)$~\cite{master}. Other properties like invariance and dimensionality of the WL algorithm have also been studied~\cite{linear_WL, inv}. Recently, graph biconnectivity based {\em Generalized 
Distance Weisfeiler-Lehman} (GD-WL)~\cite{gdwl} using distance measures have been proposed for GNN.

{\bf Motivation.} Despite its success, the WL heuristics suffers from false positive errors (i.e., graphs labeled as isomorphic might not be so) on several graph instances arising from different combinatorial 
structures~\cite{thes}. For example, strongly regular graph structures, block designs and coherent configurations provide instances where the WL test fails. Instances include the non-isomorphic $4 \times 4$ Rook's graph and Shrikhande 
graph as shown in Fig.~\ref{fig:examples}(b). Although $3$-WL test fails for the above Rook's and Shrikhande graphs, they are distinguishable by the $4$-WL test. The seminal work of Cai, 
F\"urer, and Immerman (CFI)~\cite{cfi} proved that that for every $k \geq 1$ there exists non-isomorphic 3-regular graphs of order $\mathcal{O}(k)$ that are not distinguishable by the $k$-WL algorithm. 
Further challenging examples can be considered to be incidence graphs of projective planes, Hadamard matrices, or Latin squares~\cite{traces2}. The CFI and switched edges Miyazaki construct graphs~\cite{miya} 
(see Fig.~\ref{fig:examples}(c)) provide particularly hard instances for isomorphism testing. In this work, we investigate an alternative to the WL testing heuristic that provides an enhanced discriminative power in 
terms of identifying non-isomorphic graph structure.

{\bf Contribution.} This paper proposes a novel algorithm, {\em Reachability-based Signature of Vertices using Primes} (\name), for the {\em graph isomorphism} problem. We show that \name provides an efficient 
$\mathcal{O}(m^2 + mn^2 + n^3)$ polynomial-time heuristic for isomorphism testing, and can accurately distinguish non-isomorphic graphs even for several difficult instances. Notably, we empirically show that \name can 
identify non-isomorphism in diverse challenging classes of graphs like {\em Miyazaki, Paulus, cubic hypohamiltonian, strongly regular, Latin series} and {\em Steiner triple system graphs}, wherein the $2$-WL test fails. This improves 
the scope of isomorphism detection to a broader class of graphs for better applicability, not only in terms of isomorphism testing but also for related applications like computational biology and social network analysis~\cite{ker}. 
Although our approach is also ``incomplete'' (similar to WL) and fails for extremely hard graphs like CFI, projective planar, and Hadamard matrix graphs; we hope that this work will spur interest in exploring alternatives 
to the WL heuristics. Theoretically defining the limits of our algorithm (i.e., classes of graphs where it fails) remains a future direction of study.

\section{Graph Isomorphism Testing with \name}
\label{sec:algo}

Assume $\mathcal{G}_1 = (\mathcal{V}_1, \mathcal{E}_1)$ and $\mathcal{G}_2 = (\mathcal{V}_2, \mathcal{E}_2)$ to be two input graphs (with corresponding vertex and edge sets) for isomorphism testing.
Without loss of generality, we consider both graphs to comprise $n$ vertices and $m$ edges, i.e., $|\mathcal{V}_1| = |\mathcal{V}_2| = n$ and $|\mathcal{E}_1| = |\mathcal{E}_2| = m$. Otherwise, 
the graphs can be trivially decided to be non-isomorphic as a bijection between the vertex sets would not be possible. Let $v_{ij}$ represent the $j^{th}$ vertex in the $i^{th}$ graph (i.e., $v_{ij} \in \mathcal{V}_i$) and 
$e_{kl}^i$ denote the edge $(v_{ik}, v_{il})$ in graph $\mathcal{G}_i$. Formally, the two graphs $\mathcal{G}_1$ and $\mathcal{G}_2$ are considered to be isomorphic if there exists a one-to-one mapping (i.e., bijection) 
$f: \mathcal{V}_1 \rightarrow \mathcal{V}_2$ between their vertices, such that $e_{uv}^1 (\in \mathcal{E}_1) \leftrightarrow e_{f(u)f(v)}^2 (\in \mathcal{E}_2)$.

\subsection{\centering Algorithmic Stages of \name}
\label{ssec:step}

We now describe the working steps of the proposed \name graph isomorphism testing heuristic, with Fig.~\ref{fig:algo} depicting a running example on graph $\mathcal{G}$ having $6$ vertices and $7$ edges. 

\subsection*{\indent \bf Step 1: \em Vertex Pair Distance Computation}
For each of the input graphs $\mathcal{G}_1$ and $\mathcal{G}_2$, we perform {\em breadth-first search} (BFS) to compute the hop-distance between all vertex pairs (as shown in Fig.~\ref{fig:algo}(a)). For each vertex $v$ 
in the graph, \name uses BFS (considering $v$ as the start vertex) to compute the hop distance of the other vertices from it. This pairwise vertex distances are used in the final stages of our approach to compute 
{\em vertex signatures} as discussed later in this section.

\subsection*{\indent \bf Step 2: \em Multi-Source Vertex Reachability}
Enumerating all possible paths between vertices of the graphs constitutes an exponential approach for isomorphism testing. Thus, pruning based path-finding approaches using tree constructs~\cite{path} have been explored 
in the literature. In the same spirit, \name incorporates path based ``reachability distance'' between vertices of the graph by employing a modified version of the BFS algorithm (as described next). 

\begin{figure}[t]
\centering
	\includegraphics[width=\columnwidth]{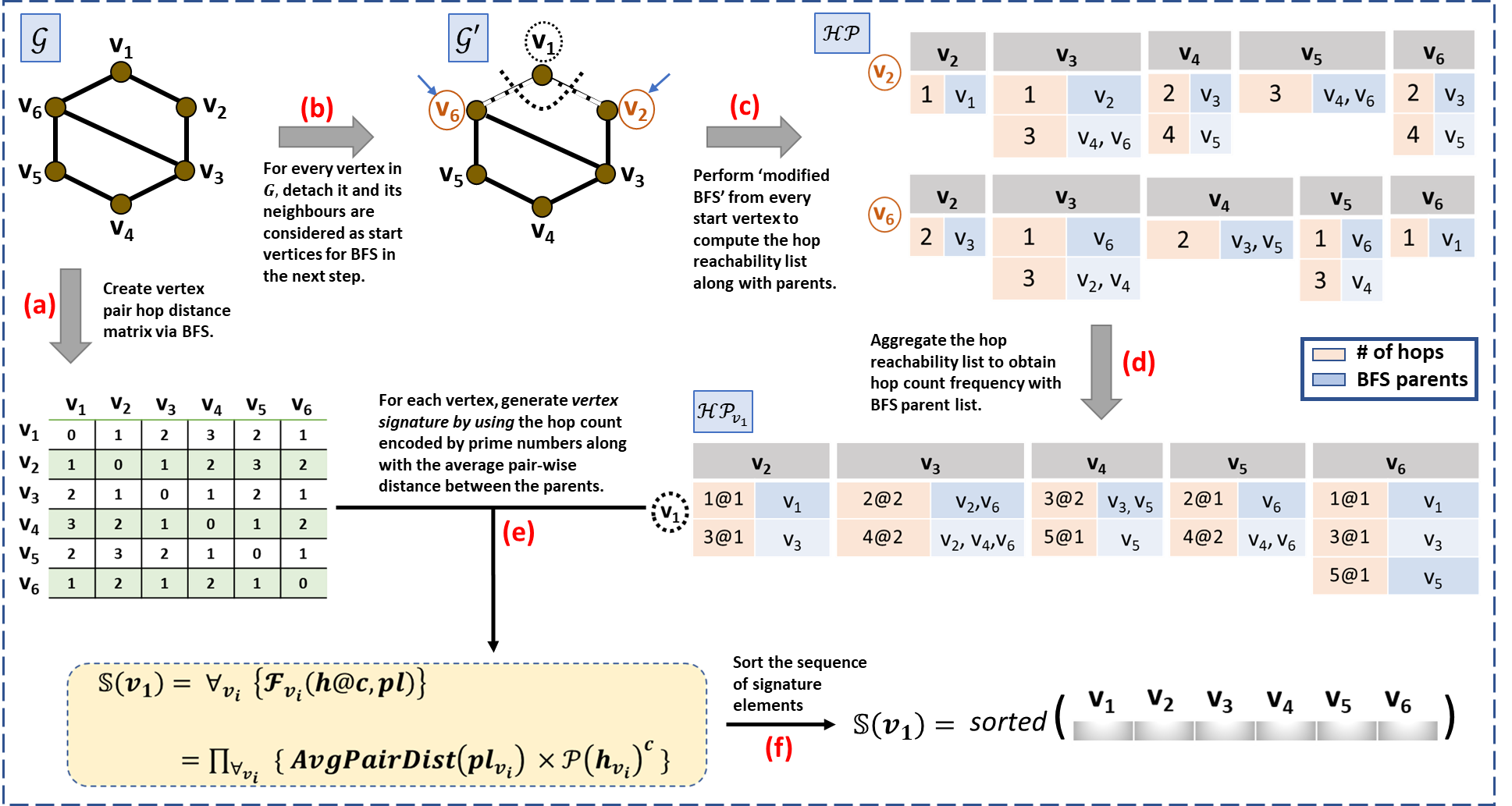}
	\caption{Running example to showcase the working steps of \name heuristic -- (a) Vertex Pair Distance Computation, (b) \& (c) Multi-Source Vertex Reachability, (d) Aggregate Vertex Hop-Reachability Distance, 
	and (e) Vertex Signature Computation.}
	\label{fig:algo}
\end{figure}

{\em Modified BFS Traversal.} For a vertex $v_1$ in graph $\mathcal{G}$, its immediate neighbours $\mathcal{N}(v_1)$ are identified and vertex $v_1$ (along with its edges) is subsequently removed from the graph, to obtain 
$\mathcal{G}'$ as shown in Fig.~\ref{fig:algo}(b). \name then performs multiple modified BFS traversal on $\mathcal{G}'$ considering each of the neighbour vertices $nv \in \mathcal{N}(v_1)$ as an independent start vertex. 
For the original vertex $v_1$, this computes several possible reachability hop-distances to the other vertices (potentially using diverse paths) based on the different start vertices. The {\em hop-distance} to a vertex along 
with its {\em immediate parent(s)} (via which it is reached during the traversal) is stored in the {\em hop-parent} index structure ($\mathcal{HP}$) as shown in Fig.~\ref{fig:algo}(c). Note, this multi-source 
(i.e., multiple start vertex) traversal enables our approach to capture different paths between vertex pairs in the graph, thereby providing important structural information to the heuristic.

Furthermore, during the ``modified'' BFS traversal, although each vertex of the graph is {\em discovered} and {\em visited}
\footnote{We consider a vertex to be ``discovered'' when it is pushed into the BFS heap from a ``visited'' neighbour. A vertex is said to be ``visited'' only when all its incident edges have been considered and the vertex is 
subsequently removed from the heap.\label{foot}} 
only once (i.e., inserted only once in the associated heap), edges to visited vertices from a discovered vertex are {\bf not} ignored. That is, if one of the end vertices of an edge is not visited, then the corresponding 
edge is considered as a possible new path from the unvisited vertex. For example, in graph $\mathcal{G}'$ (Fig.~\ref{fig:algo}(b)), visiting of vertex $v_3$ leads to the discovery (and heap insertion) of vertex 
$v_4$\footref{foot}. Subsequently during the ``visiting process'' of $v_4$, the edge $(v_4, v_3)$ is considered to capture a potential new path to $v_3$ (since $v_4$ is not yet visited although $v_3$ has 
been visited) -- creating a corresponding entry in index $\mathcal{HP}$. Observe, this might lead to an edge being traversed multiple times, albeit a constant number of times. 

We refer to this process of ``modified'' BFS traversal from multiple start vertices incorporating different paths between vertices of $\mathcal{G}$ as {\em multi-source vertex reachability} computation -- enabling enhanced 
capture of structural information in \name (compared to colour refinement procedure). For example, in Fig.~\ref{fig:algo}(b), we are thus able to capture both the $2$-hop and $4$-hop paths between $v_2$ and $v_4$ (i.e., 
$v_2 \rightarrow v_3 \rightarrow v_4$ and $v_2 \rightarrow v_3 \rightarrow v_6 \rightarrow v_5 \rightarrow v_4$), even though $v_6$ might have been visited earlier. This information is duly represented for $v_4$ in $\mathcal{HP}$ as 
$\langle 2, v_3 \rangle$ and $\langle 4,v_5 \rangle$ corresponding to the BFS start vertex $v_2$ (Fig.~\ref{fig:algo}(c)). 

The above process of multi-source vertex reachability computation is performed and aggregated for every vertex in the input graphs $\mathcal{G}_1$ and $\mathcal{G}_2$, as discussed next.

\subsection*{\indent \bf Step 3: \em Aggregate Vertex Hop-Reachability Distance}
Fig.~\ref{fig:algo}(c) depicts the {\em hop-parent} ($\mathcal{HP}$) indices for vertex $v_1$ in graph $\mathcal{G}$, obtained by the ``modified BFS'' procedure from its neighbouring vertices $\mathcal{N}(v_1) = \{v_2, v_6\}$.
Subsequently, these $\mathcal{HP}$ indices of the above multi-source vertex reachability step are aggregated for the original vertex $v_1$, to obtain $\mathcal{HP}_{v_1}$. To this effect, $\mathcal{HP}_{v_1}$ 
stores the $\langle${\em hop-distance:\{parent list\}}$\rangle$ from $v_1$ to every other vertex in $\mathcal{G}$. 

Let $\mathcal{HP}_{v_1 \rightarrow v_t} = \{\langle h@c_{v_1 \rightarrow v_t}, pl_{v_1 \rightarrow v_t} \rangle\}$ denote the collection of $\langle${\em hop-distance:\{parent list\}}$\rangle$ for vertex $v_t \in \mathcal{G}$ 
reachable from $v_1$. Here, $h$ represents the {\em hop-distances} of $v_t$ (i.e., the hops required to reach it) from $v_1$. Additionally, $c$ denotes the number of unique immediate parents via which $v_t$ was reachable from $v_1$ 
with $h$ hops, while $pl_{v_1 \rightarrow v_t}$ enumerates these parent vertices of $v_t$ (as shown in Fig.~\ref{fig:algo}(d)). Formally, for graph $\mathcal{G}$ we have $\mathcal{HP}_{v_i \in \mathcal{G}} = 
\forall_{v_j \in \mathcal{G}} \{ \langle h@c_{v_i \rightarrow v_j}, pl_{v_i \rightarrow v_j} \rangle \}$. To tackle corner cases, disconnected and unreachable vertices would have a hop-count of $0$ and $\emptyset$ parent list.

The $\mathcal{HP}_{v_i}$ index is similarly constructed for all the vertices of the input graphs $\mathcal{G}_1$ and $\mathcal{G}_2$, and is next used to compute the {\em signatures} of the vertices.

\subsection*{\indent \bf Step 4: \em Vertex Signature Generation}
Given the above aggregated $\mathcal{HP}_{v_i}$ structure for vertex $v_i$ in graph $\mathcal{G}$, \name next constructs the vertex {\em signature}, $\mathbb{S}(v_i)$, using prime encoding of the hop distances 
and the average distance between the BFS parents of the reachable vertices. Formally, we compute the vertex signature as,
\begin{align}
	\mathbb{S}(v_i) = \mathcal{F}(\mathcal{HP}_{v_i}) = \forall_{v_j \in \mathcal{G}} ~\mathcal{F}(\{ \langle h@c_{v_i \rightarrow v_j}, pl_{v_i \rightarrow v_j} \rangle \})
\end{align}
where $\mathcal{F}:  \langle h@c, pl \rangle \rightarrow \mathbb{R}$. Thus, the {\em signature} of $v_i$ can be viewed as 
as a {\em sequence of $n$ real-numbered values} obtained from $\mathcal{F}$ for each of the $n$ vertices of the graph (see Fig.~\ref{fig:algo}(e)).

Consider $\mathbb{P}_{\geq 3}$ to be the set of prime numbers greater than $2$. We define a unique mapping $\mathcal{P}: \mathbb{N} \rightarrow \mathbb{P}$, which is used by \name to transform a hop count value ($h$) to a 
prime number. This {\em prime encoding of the vertex hop reachability} forms the fundamental strategy for isomorphism testing in \name.
Hence, the $t^{th}$ element in the signature of $v_1$ corresponding to a reachable vertex $v_t$ (for $v_1, v_t \in \mathcal{G}$) is,
\begin{align}
\label{eq:sign}
	\mathcal{S}(v_i \rightarrow v_t) &= \mathcal{F}(\mathcal{HP}_{v_i \rightarrow v_t}) = \mathcal{F}\Bigl( \bigl\{ \langle h@c_{v_i \rightarrow v_t}, pl_{v_i \rightarrow v_t} \rangle \bigr\} \Bigr) \nonumber \\
	&= \prod_{\text{for each $\langle h@c, pl \rangle$}} \bigl\{AvPD(pl) \times \mathcal{P}(h)^c \bigr\}
\end{align}
where $AvPD$ is the average pairwise distance between the vertices in the parent list $pl$.

Intuitively, the ``modified'' BFS along with the {\em average pairwise parent distance} provides vital structural information. For example, in strongly 
regular graphs (e.g., Shrikhande and $4 \times 4$ Rook's graphs in Fig.~\ref{fig:examples}(b)) most vertices have the same degree and are reachable via the same number of hops from other vertices. On the other hand, 
the Rook's graph possesses 4-cliques (not present in the Rook's graph), while only the Shrikhande graph features 5-rings~\cite{cw}. We hypothesize that such differences in sub-structural properties are captured 
by \name enabling it to distinguish between these non-isomorphic graphs.

For our example in Fig.~\ref{fig:algo}(e), the $3^{rd}$ element of the signature of $v_1$ corresponds to vertex $v_3$ and is computed by Eq.~\eqref{eq:sign} based on the $\langle h@c, pl \rangle$ entries in 
$\mathcal{HP}_{v_1 \rightarrow v_3}$. Thus, we have, 
\begin{align*}
	\mathcal{S}(v_1 \rightarrow v_3) &= \mathcal{F}(\mathcal{HP}_{v_1 \rightarrow v_3}) = \mathcal{F}\Bigl( \bigl\{ \langle 2@2, \{v_2, v_6\} \rangle, \langle 4@2, \{v_2,v_4,v_6\} \rangle \bigr\} \Bigr) \\
	&= \mathcal{F}\bigl(\langle 2@2, \{v_2, v_6\} \rangle \bigr) \times \mathcal{F}\bigl(\langle 4@2, \{v_2,v_4,v_6\} \rangle \bigr) \\
	&= \Bigl[ \mathcal{P}(2)^2 \times AvPD\bigl(\{v_2, v_6\}\bigr) \Bigr] \times \Bigl[ \mathcal{P}(4)^2 \times AvPD\bigl(\{v_2, v_4, v_6\}\bigr) \Bigr]
\end{align*}

Thus for vertex $v_i$, the collection of signature elements (as in Eq.~\eqref{eq:sign}) pertaining to the other vertices of the graph forms the sequence $\Bigl\{ \forall_{v_j \in \mathcal{G}}~\mathcal{S}(v_i \rightarrow v_j)\Bigr\}$. 
This sequence is then sorted to obtain the {\em signature} of vertex $v_i$, i.e., $\mathbb{S}(v_i) = sorted \Bigl( \bigl\{ \forall_{v_j \in \mathcal{G}}~\mathcal{S}(v_i \rightarrow v_j)\bigr\} \Bigr)$, as shown in 
Fig.~\ref{fig:algo}(f).

The vertex signatures for all the vertices of the input graphs are similarly computed by \name.

\subsection*{\indent \bf Step 5: \em Graph Representation and Isomorphism Testing}
Finally, for an input graph $\mathcal{G}$, \name constructs its representation or ``certificate'' as a collection of signatures of its vertices. Formally, $\mathscr{C}(\mathcal{G)} = \bigl\{\mathbb{S}(v_i)\bigr\}$, for all 
$v_i \in \mathcal{G}$. For isomorphism testing between two input graphs $\mathcal{G}_1$ and $\mathcal{G}_2$, the generated ``certificates'' of the graphs are compared, as described next.

Let $\mathscr{C}(\mathcal{G}_1) = \bigl\{\forall_{i=1}^n ~\mathbb{S}(v_{1i})\bigr\}$ and $\mathscr{C}(\mathcal{G}_2) = \bigl\{\forall_{j=1}^n ~\mathbb{S}(v_{2j})\bigr\}$ be the certificates for graphs $\mathcal{G}_1$ 
and $\mathcal{G}_2$, respectively, with $n$ vertices. If $\exists v_{1p} \in \mathcal{G}_1, ~v_{2q} \in \mathcal{G}_2 ~\mid ~\mathbb{S}(v_{1p}) = \mathbb{S}(v_{2q})$, for $\mathbb{S}(v_{1p}) \in \mathscr{C}(\mathcal{G}_1), 
~\mathbb{S}(v_{2q}) \in \mathscr{C}(\mathcal{G}_2)$, \name considers the vertices to be a \emph{bijection} (i.e., 1-to-1 mapping) and the vertices are, thus, removed from $\mathscr{C}(\mathcal{G}_1)$ and $\mathscr{C}(\mathcal{G}_2)$ respectively. 
If there is a mismatch and no such matching candidate vertex signature is found between the graph certificates, we report $\mathcal{G}_1$ and $\mathcal{G}_2$ to be {\em non-isomorphic}. On the other hand, if a bijection 
(as described above) is found between all the vertex signatures of the input graphs, we consider the graphs to be {\em isomorphic} (and the bijection between the vertices are returned).

\subsection{\centering Correctness Analysis of \name}
\label{ssec:cor}

\begin{theorem}
\label{th:prime}
	For two {\em isomorphic} graphs $\mathcal{G}_1$ and $\mathcal{G}_2$, if there exists a {\em bijection} between vertices $v_{1i} \in \mathcal{G}_1$ and $v_{2j} \in \mathcal{G}_2$, then the {\em signatures} of $v_{1i}$ and 
	$v_{2j}$ computed by \name are {\em identical}, i.e., $\mathbb{S}(v_{1i}) = \mathbb{S}(v_{2j})$.
\end{theorem}
\begin{proof}
	Since the vertices $v_{1i} \in \mathcal{G}_1$ and $v_{2j} \in \mathcal{G}_2$ form a bijection between the isomorphic graphs $\mathcal{G}_1$ and $\mathcal{G}_2$, they are structurally identical. That is, they have the 
	same structural properties in terms of degree, neighbours, paths to other vertices and so on. Consider the two vertex features that \name uses for computing the vertex signatures, namely (i) {\em prime encoding of 
	hop-distance} to other vertices and (ii) {\em average pairwise parent distance} for the vertices.
	
	Note, the function $\mathcal{P}$ uniquely maps each hop count value to a prime, and the vertex signature involves the product of these prime encodings for the reachable vertices. Since prime numbers do not have any 
	{\em proper factors} (excluding $1$ and itself), the product of the primes are unique for a given combination of the hop counts. Thus, for a vertex $v_{1t} \in \mathcal{G}_1$, the corresponding $t^{th}$ element in 	$\mathbb{S}(v_{1i})$, 
	the signature of $v_{1i}$, contains the \emph{unique} product of the prime encodings of the hop counts via which $v_{1t}$ can be reached from $v_{1i}$. 

	Since, the input graphs are isomorphic, we consider $v_{1t}$ to have a bijection to vertex $v_{2s} \in \mathcal{G}_2$. Now, for vertex $v_{2j}$, the $s^{th}$ element in $\mathbb{S}(v_{2j})$ similarly contains the 
	product of prime encodings of the hop counts for reaching $v_{2s}$ from $v_{2j}$. Since, $v_{1i}$ and $v_{2j}$ are structurally identical, the hop counts to $v_{1t}$ and $v_{2s}$ respectively are the same and, thus, have the 
	same product of hop count prime encodings. 

	Further, the parent list for $v_{1t}$ and $v_{2s}$ also consists of equivalent vertices (as the graphs are isomorphic) and, hence, their average pairwise distance would also be the same. Combining this with the unique product of 
	hop primes (via function $\mathcal{F}$), the signature (after the sorting step) of vertices $v_{1i}$ and $v_{2j}$ would be identical (i.e., $\mathbb{S}(v_{1i}) = \mathbb{S}(v_{2j})$) if they share a bijection in the isomorphic 
	input graphs. 
\end{proof}

\begin{corollary}
\label{cor:iso}
	If two graphs $\mathcal{G}_1$ and $\mathcal{G}_2$ are {\em isomorphic}, their ``certificates'' or representations computed by \name are {\em equivalent}, i.e., $\mathscr{C}(\mathcal{G}_1) \equiv \mathscr{C}(\mathcal{G}_2)$.
\end{corollary}
\begin{proof}
	The certificate of an input graph is constructed from the signatures of the vertices of the graph. For isomorphic graphs, from Theorem~\ref{th:prime}, for every $v_{1i} \in \mathcal{G}_1$ there exists $v_{2j} \in \mathcal{G}_2$ 
	having an identical signature. These are reported as the bijection between the vertices for the isomorphic input graphs. Thus, there exists an ordering of the representations of such input graphs that are identical. To this end, 
	isomorphic graph certificates constructed by \name are equivalent. 
\end{proof}

\noindent {\bf Observation.} The proposed \name graph isomorphism testing is ``incomplete'' in the same sense as that of the Weisfeiler-Lehman procedure. \name thus suffers from one-sided error, where isomorphic graphs will 
never be deemed as non-isomorphic (as established in Theorem~\ref{th:prime}), although the reverse might occur. In other words, if two graphs have different representations they are definitely non-isomorphic, but graphs with 
identical representations may not be isomorphic. Hence, non-isomorphic graphs might possibly have identical graph certificates in \name, and can be classified as isomorphic. 

Observe, although the 
``modified'' BFS attempts to capture structural information of the graphs, it would fail to detect long paths between vertices, which might involve already visited intermediate vertices. Due to this partial nature of the 
structural information in \name, it might fail to distinguish non-isomorphic graphs. However, in practice, we observe that \name is quite effective in detecting non-isomorphism even for a wide range of challenging 
graph classes (refer Sec.~\ref{sec:expt}).

\subsection{\centering Runtime Complexity of \name}
\label{ssec:run}

\begin{table}[t]
\centering
\scriptsize
\caption{Comparative study of performance of \name and Weisfeiler-Lehman (WL) on {\em \bf non-isomorphic} instances of hard graph classes. We observe that \name achieves better performance than WL in detecting non-isomorphism 
for majority of the challenging scenarios. For a few graph classes \name also exhibits false positive results.}
\label{tab:non_iso}
	\begin{tabular}{c||c|c}
		\toprule
		{\bf Graph Class / Algorithm} & {\bf Weisfeiler-Lehman Test} & {\bf \name Heuristic} \\
		\midrule
		\midrule
		{\em Disconnected property (Fig.~\ref{fig:examples}(a))} & \circledcross & \circledcheck \\
		{\em Strongly regular property (Fig.~\ref{fig:examples}(b))} & \circledcross & \circledcheck \\
		\midrule
		{\em Switched Edge Miyazaki I and II Graphs (Fig.~\ref{fig:examples}(c))} & \circledcross & \circledcheck \\
		{\em Random Regular Graphs (RND-3-REG)} & \circledcross & \circledcheck \\
		{\em Union of Strongly Regular Graphs (USR)} & \circledcross & \circledcheck \\
		{\em Non-Disjoint Union of Tripartite Graphs (TNN)} & \circledcross & \circledcheck \\
		{\em Cubic Hypohamiltonian Graphs (CHH)} & \circledcross & \circledcheck \\
		{\em Steiner Triple System Graphs with switched edges (STS-SW)} & \circledcross & \circledcheck \\
		{\em Product Graphs (F-LEX)} & \circledcross & \circledcheck \\
		{\em Latin-SW Series Graphs} & \circledcross & \circledcheck \\
		{\em Paulus Graphs} & \circledcross & \circledcheck \\
		\midrule
		{\em Dawar-Yeung construction on Cai, F\"urer and Immerman Graphs (SAT-CFI)} & \circledcheck & \circledcheck \\
		\midrule
		{\em CFI I and II} & \circledcross & \circledcross \\
		{\em Projective Plane Graphs (PP)} & \circledcross & \circledcross \\
		{\em Hadamard Matrix Graphs with switched edges (HAD-SW)} & \circledcross & \circledcross \\
		{\em Multipedes (on cyclic group $\mathbb{Z}_2$) (z2)} & \circledcross & \circledcross \\
		{\em Dihedral Construction $R(B(G_n,\sigma)$) (on dihedral group $D_3$) (d3)} & \circledcross & \circledcross \\
		{\em Rigid Base Construction $R^*(B(G_n,\sigma)$) (on cyclic group $\mathbb{Z}_2$) (s2)} & \circledcross & \circledcross \\
		{\em Rigid Base Construction $R(B(G_n,\sigma)$) (on cyclic group $\mathbb{Z}_3$) (z3)} & \circledcross & \circledcross \\
		\bottomrule
	\end{tabular}
\end{table}

\begin{theorem}
\label{th:run}
	The runtime complexity of \name is {\em polynomial} and is bounded by $\mathcal{O}(m^2 + mn^2 + n^3)$, where $n$ and $m$ are the number of vertices and number of edges, respectively, in the input graphs.
\end{theorem}
\begin{proof}
Consider that both the input graphs for isomorphism testing contain $n$ vertices and $m$ edges. 

The {\em vertex pair distance computation} (Step 1 in Sec.~\ref{ssec:step}) performs a breadth-first search (BFS) from each of the vertices of the graph. Thus this process takes 
$n \times \mathcal{O}(m + n) = {\bf \mathcal{O}(mn + n^2)}$ runtime. 

In Step 2 we construct the multi-source vertex reachability index, $\mathcal{HP}$, for each of the vertices. Here, for vertex $v$ with degree $deg(v)$, \name performs $deg(v)$ ``modified'' BFS traversals, each starting 
from one of the neighbours of $v$. Note, in our ``modified`` BFS procedure, although an edge can be visited multiple times, each vertex is visited (and inserted in the heap) only {\em once}, as discussed in Sec.~\ref{sec:algo}. 
Observe that an edge ($u,v$) is actually traversed {\em twice} -- once during visiting $u$ (during which $v$ is discovered) and then again when $v$ is visited. After both $u$ and $v$ have been visited, they are not 
inserted into the heap any further and, hence, edge $(u,v)$ is not encountered again. Thus, the time complexity of the ``modified BFS'' remains $\mathcal{O}(m + n)$. Combining the above, for each vertex construction of the 
reachability index takes $deg(v) \times \mathcal{O}(m + n)$ time. Hence, the overall runtime complexity of creating $\mathcal{HP}$ for all vertices (Step 2 in \name) is bounded by $\sum_{\forall v_i} deg(v_i) 
\times \mathcal{O}(m + n) = \mathcal{O}(m) \times \mathcal{O}(m + n) = {\bf \mathcal{O}(m^2 + mn)}$.

The aggregation of hop-reachability distance for vertex $v$ (in Step 3) involves merging the above $deg(v)$ instances of $\mathcal{HP}$ structures pertaining to the neighbours of $v$ (considered as the start vertex for 
``modified'' BFS). Further, each $\mathcal{HP}$ contains $\bigl\{\langle$ hop-distance:\{parent list\} $\rangle\bigr\}$ entries for every $v_t \in \mathcal{G}$ (i.e., $\mathcal{HP}_{v \rightarrow v_t}$). Note, hop-distance 
is bounded by $\mathcal{O}(n)$, while the size of the parent list is at most $deg(v_t)$. Thus, for every vertex $v_t$, \name merges $deg(v)$ lists (from each of the $\mathcal{HP}$ associated with original vertex $v$) 
with size $\mathcal{O}(n)$, thereby leading to a time complexity of $\sum_{\forall_{v_t \in \mathcal{G}}} ~deg(v) \times \mathcal{O}(n) = deg(v) \times \mathcal{O}(n^2)$. Finally, in Step 3, \name merges $\mathcal{HP}$ indices 
pertaining to all vertices of the input graph, thus having a runtime complexity of $\sum_{\forall_{v \in \mathcal{G}}} deg(v) \times \mathcal{O}(n^2) = {\bf \mathcal{O}(mn^2)}$.

The {\em vertex signatures} are computed in Step 4 based on Eq.~\eqref{eq:sign}. For each vertex $v$, the prime encoding of the hop counts takes $\mathcal{O}(n)$ (as hop count is bounded by the number of vertices), 
while computing the average pairwise distance between the parents is at most $\mathcal{O}(n^2)$ (bounded by the access to all elements of the vertex pair distance matrix of Step 1). Thus, the sequence of $n$ real-valued 
signature vector of $v$ computed using $\mathcal{F}$ and its subsequent sorting is performed in $\mathcal{O}(n^2 + n \log n)$ time. So, generating the graph ``certificate'' with all the vertex signatures (Step 5) has a time complexity 
of $n \times \mathcal{O}(n^2 + n \log n) = {\bf \mathcal{O}(n^3)}$.

Finally, for isomorphism testing, \name compares the vertex signatures between the input graphs to find a bijection between the $2n$ vertices, which takes 
${\bf \mathcal{O}(n^2)}$ time.

Combining the above, we observe that the runtime complexity of \name for isomorphism testing is bounded by $\mathcal{O}(mn + n^2) + \mathcal{O}(m^2 + mn) + \mathcal{O}(mn^2) + \mathcal{O}(n^3) + \mathcal{O}(n^2) = 
{\bf \mathcal{O}(m^2 + mn^2 + n^3)}$, i.e., {\em polynomial} in terms of the number of edges and vertices of the input graphs\footnote{This translates to $\mathcal{O}(n^3)$ for sparse graphs ($\mathcal{O}(m) 
\approx \mathcal{O}(n)$), and to $\mathcal{O}(m^2 + mn^2)$ for dense graphs.}.
\end{proof}

As comparison, note that the $3$-WL algorithm having a runtime of $\mathcal{O}(n^4 \log n)$ is unable to distinguish between the non-isomorphic Shrikhande and $4\times4$ Rook's graphs, while \name correctly reports them. 
In fact, distance based WL (using biconnectivity) approach also fails in this respect~\cite{gdwl}. Although $4$-WL is able to distinguish the above, it suffers from a prohibitive runtime of $\mathcal{O}(n^5 \log n)$. Thus, 
\name provides an efficient and more effective alternative to the WL approach for graph isomorphism.

\section{Experimental Evaluation}
\label{sec:expt}

\begin{table}[t]
\centering
\caption{Comparative study of performance of \name and Weisfeiler-Lehman (WL) on {\em \bf isomorphic} instances of hard graph classes. Both \name and WL are able to detect isomorphism for 
all the graph classes. This demonstrates the property of one-sided error for both the heuristics.}
\label{tab:iso}
	\begin{tabular}{c||c|c}
		\toprule
		{\bf Graph Class / Algorithm} & {\bf Weisfeiler-Lehman Test} & {\bf \name Heuristic} \\
		\midrule
		\midrule
		{\em AG Series} & \circledcheck & \circledcheck \\
		{\em CMZ Series} & \circledcheck & \circledcheck \\
		{\em Grid Series} & \circledcheck & \circledcheck \\
		{\em KEF Series} & \circledcheck & \circledcheck \\
		{\em MZ Series} & \circledcheck & \circledcheck \\
		{\em Latin Series} & \circledcheck & \circledcheck \\
		{\em Lattice Series} & \circledcheck & \circledcheck \\
		{\em PG Series} & \circledcheck & \circledcheck \\
		{\em k-Complete Series} & \circledcheck & \circledcheck \\
		{\em Triang Series} & \circledcheck & \circledcheck \\
		{\em Paley Series} & \circledcheck & \circledcheck \\
		{\em Shrunken Multipedes (on cyclic group $\mathbb{Z}_2$) (t2)} & \circledcheck & \circledcheck \\
		{\em Rigid Base Construction $R(B^*(G_n,\sigma)$) (on cyclic group $\mathbb{Z}_2$) (r2)} & \circledcheck & \circledcheck \\
		\bottomrule
	\end{tabular}
\end{table}

In this section, we present empirical evidence as to the efficacy of the proposed \name isomorphism testing heuristic. To this end, we evaluate the discriminative power of \name for several challenging and hard families 
of graph instances, like {\em Miyazaki, Paulus, cubic hypohamiltonian, strongly regular, Latin series} and {\em Steiner triple system graphs} to name a few. We compare our approach to the classical and widely used 
$2$-Weisfeiler-Lehman ($2$-WL) approach, and report the results for both isomorphic and non-isomorphic graph instances.

We evaluated the approaches over a wide gamut of graph classes across several isomorphism testing benchmark datasets as obtained from the following {\em three} sources: \\
(i) \url{https://pallini.di.uniroma1.it/Graphs.html}; \\
(ii) \url{https://www.lics.rwth-aachen.de/go/id/rtok/}; and, \\
(iii) \url{http://www.tcs.hut.fi/Software/bliss/benchmarks/index.shtml} \\
The Miyazaki MZ I and MZ II graphs were obtained from \url{https://codeocean.com/capsule/8055748}~\cite{quan}.

Table~\ref{tab:non_iso} compares the performance of \name and WL approaches in identifying non-isomorphic graph classes. We observe that \name demonstrates far lesser number of false positives (i.e., non-isomorphic 
graphs classified as isomorphic) compared to the WL test. Specifically, for the $19$ hard non-isomorphic graph classes shown in Table~\ref{tab:non_iso}, \name was able to correctly distinguish $12$ (i.e., $\sim 63\%$) of them, 
while WL could detect only $1$ (i.e., $\sim 5\%$). This marked improvement in performance for ``hard'' instances would be beneficial not only for isomorphism detection but also for other applications and domains.

However, we do observe (in Table~\ref{tab:non_iso}) that \name mistakenly reports certain classes of non-isomorphic graphs as isomorphic. This includes CFI graphs, Projective Planar and Hadamard Matrix graphs among others. We 
believe that the layered structure of these graphs with many vertices having equal degree makes it difficult for our ``modified'' BFS traversal to suitably capture sufficient structural information. In fact, these graphs are 
considered among the possibly hardest instances for graph isomorphism detection.

Finally, for completeness, Table~\ref{tab:iso} also reports the performance of \name and WL approaches in identifying isomorphic graphs, wherein both the approaches are perfect (i.e., theoretically proven to be accurate) in 
this respect.

\section{Discussion and Conclusion}
\label{sec:conc}

This work presents \name, a polynomial time graph isomorphism testing heuristic. On various hard graph class instances (e.g., Miyazaki, strongly regular and cubic hypohamiltonian graphs), we depicted that our approach could 
accurately distinguish non-isomorphic graphs, whereas the well-known WL test failed, thus significantly expanding the scope of polynomial time isomorphism detection. For finding isomorphism, \name computes {\em vertex 
signatures} capturing structural information based on {\em prime encoding} of ``reachability distance'' to other vertices coupled with the average pair-wise distance among the parents of a vertex.

Although \name suffers from ``one-sided'' error and fails for certain extremely hard graphs (e.g., CFI and projective planar graphs), this work would benefit downstream applications of graph isomorphism as well as hopefully 
spur renewed interest in exploring alternatives to WL testing procedure. 

The exact nature or characteristics of the graphs class on which \name is guaranteed to work (or fail) is still unknown, and provides an interesting direction of future work. Possibly this could lead to the formulation 
of a new class of graphs, the isomorphism of which can be provably solved in polynomial time. In fact, since the current known theoretical bound for GI by Babai~\cite{bound2a} uses $k$-WL heuristic, it would be interesting 
to study if replacing it with our approach provides further improvements.

\bibliographystyle{IEEEtran}
\bibliography{ref}

\end{document}